\documentclass{amsart}
\usepackage{graphicx} 
\usepackage{amssymb}
\usepackage{amsfonts}
\usepackage{bm}
\usepackage{graphicx}
\usepackage{epsfig}
\usepackage{float}
\usepackage{caption}
\usepackage{tikz}

\newtheorem{thm}{Theorem}[section]
\newtheorem{prop}[thm]{Proposition}
\newtheorem{cor}[thm]{Corollary}
\newtheorem{lemma}[thm]{Lemma}

\newtheorem*{HLK}{Theorem (Huygens-Leibniz-K\"{o}nig)}
\newtheorem*{ACF}{Corollary (Albouy-Chenciner)}

\theoremstyle{definition}
\newtheorem{defn}[thm]{Definition}

\theoremstyle{definition}

\newtheorem{rem}[thm]{Remark}

\newcommand{\ssp}{\mathbb{W}}

\newcommand{\reals}{\mathbb{R}}

\newcommand{\longvec}{\overrightarrow}
\newcommand{\vecaffX}{\overrightarrow{\text{Aff}(X)}}
\newcommand{\affspa}{(\mathbb{A},\longvec{\mathbb{A}},\tau)}

\newcommand{\aff}{\text{Aff}}
\newcommand{\bary}{\text{bar}}

\numberwithin{equation}{section}

\title{Moments, Equilibrium Equations and Mutual Distances}
\author{Eduardo S. G. Leandro}
\address{Universidade Federal de Pernambuco (UFPE) \\
	Departamento de Matem\'atica \\
	Av. Jornalista An\'{\i}bal Fernandes S/N, Recife, PE \\
	50740-560, Brasil\\
	Phone: 55 81 2126-7665, Fax: 55 81 2126-8410.
}
\email{eduardo.leandro@ufpe.br}


\begin{document}
	

	\begin{abstract}
		We review and develop the classical theory of moments of configurations of weighted points with a focus on systems with an identically vanishing first moment. The latter condition produces equations for equilibrium configurations of systems of interacting particles under the sole condition that interactions are between pairs of particles and along the lines connecting such pairs. Complying external forces are admitted, so the description of some dynamical equilibrium configurations, such as relative equilibria in Celestial Mechanics, is included in our approach. Moments provide a unified framework for equilibrium problems in arbitrary dimensions. The equilibrium equations are homogeneous and invariant by isometries (for interactions depending only on mutual distances), and are obtained through simple algebraic procedures requiring neither reduction by isometries nor a variational principle for their determination.  Our equations include the renowned set of $n$-body central configuration equations by A. Albouy and A. Chenciner. These equations are extended to a rather broad class of equilibrium problems, and new equilibrium equations written in terms of mutual distances are introduced. We also apply moments to the theory of constraints for mutual distances of configurations of fixed dimension, and for co-spherical configurations, thus re-obtaining and adding to classical results by A. Cayley and successors.	For the sake of concreteness, novel sets of central configurations equations are provided.	
	\end{abstract}

	\maketitle
	
	\section{Introduction} \label{Intro}
	Equilibrium problems have contributed to the development of Mathematics and Physics since ancient Greece. The work of Archimedes on the principle of the lever is considered the beginning of the mathematical treatment of Statics, but even before him, there had appeared books attributed to Aristotle and to Euclid on the dynamics and balance of the lever, respectively~\cite{Dugas}. The concept of barycenter, or center of gravity, of a system of point masses, and its application to studying equilibrium configurations go back to Archimedes. As we shall review in section~\ref{secWS}, the barycenter is defined for sets of weighted points whose total weight, $\mu_0$, is nonzero, as the unique point where the first moment, $\mu_1$, vanishes. In the 17th century, C. Huygens used $\mu_1$ to introduce the notion of expected value in probability, and realized the importance of the second moment, $\mu_2$, in the form of the moment of inertia (Euler's terminology), for the analysis of pendulum motion. Later on G. W. Leibniz and J. S. K\"{o}nig applied $\mu_2$ in Geometry and Dynamics~\cite{Gidea-Niculescu}. A consequential event was the invention by A. F. M\"{o}bius of homogeneous, barycentric coordinates in the early 19th century~\cite{Coxeter}. In contemporary Mathematical Statistics and Analysis, among other areas, one finds the celebrated (Stieltjes, Hausdorff, Hamburger) moment problem: given a sequence of moments of a real random variable, determine the (unique) probability distribution (measure) that produced the sequence~\cite{Landau,Shohat-Tamarkin}. 
	
	In the physical sciences, it has been a long tradition to represent configurations of (interacting) point masses, electric charges, point vortices, etc, by a finite set $X$ formed by $n$ points in the Euclidean affine space $\mathbb{A}=\reals^N$ over the field of scalars $K=\reals$, and a  function $w$ from $X$ to $K$ denoting the mass, charge, vorticity, etc, of each point of $X$. In an abstract sense, we refer to $X$ and $w$ as \emph{configuration} and \emph{weight function}, respectively, and call the pair $(X,w)$ a \emph{weighted system}. 
	Moving beyond this tradition, in this paper we view the interactions of particles as themselves defining the weight functions. More precisely, we introduce a family of weighted systems $(X,w_x)$, $x \in X$, associated with a configuration $X$ of interacting particles, so that $w_x(y)$ is the \emph{force} (this familiar term will be often used in place of \emph{interaction}) exerted by the particle $y \in X$ on the particle $x$, while $w_x(x)$ is the negative of the total force acting on $x$. 
	
	As already mentioned, three  elementary functions are historically associated with every weighted system: the total weight $\mu_0$ (``zeroth moment''), the first moment $\mu_1$ and the second moment $\mu_2$.  The key realization that an equilibrium configuration $X$ corresponds to the family of weighted systems $(X,w_x)$, $x \in X$, fulfilling the condition $\mu_1=\longvec{O}$ (this notation means that the first moment is the null function) for each system $(X,w_x)$ motivates a deeper investigation of the theory of moments of weighted systems in affine spaces. Moments allow us to formulate problems of equilibrium configurations of systems of interacting particles (possibly subject to certain types of external forces) that satisfy the single fundamental hypothesis
	{\bf(FH)}: \emph{the interaction between the particle at $x$ and the particle at $y$, with $x,y \in X$, is directed along the line determined by $x$ and $y$}. 
	\label{fundamentalhypothesis} Actually it suffices that the interactions fulfill the weaker fundamental hypothesis {\bf(WFH)}: \emph{the total interaction (e.g., the sum of all internal and external forces) at each point $x$ of the configuration $X$ is a linear combination of the relative position vectors from $x$ to the remaining points of $X$}. \label{fundamentalhypothesisI} Hypothesis {\bf(WFH)} extends the scope of applications of the moment apparatus to systems affected by some inertial forces, for instance, centrifugal forces giving rise to relative equilibrium configurations.

	The main pillar for the present article are the classical identities deduced by the pioneers of moment theory and their use, explicitly or not, to study central configurations of the $n$-body problem. 
	
	Central configurations is a research topic with a long history and recognized relevance to Celestial Mechanics~\cite{Albouy-Chenciner,Albouy,Moeckel,Smale,Wintner}. One of the reasons is that the only known explicit solutions of the  Newtonian $n$-body problem, namely the \emph{homographic motions}, start at such configurations (cf. Wintner~\cite{Wintner}, \S375). The possible homographic motions include the \emph{relative equilibria}, which are periodic solutions that occur in any Euclidean space of even dimension~\cite{Albouy-Chenciner,Moeckel}. 
	Notwithstanding the many particular questions that have been answered, so far a classification of central configurations is known only for three arbitrary masses, thanks to Euler~\cite{Euler} and Lagrange~\cite{Lagrange}, and for four equal masses, thanks to Albouy~\cite{AlbouyIII}.  General, basic problems on central configurations remain unsolved. An outstanding example is the celebrated Smale's sixth problem for the 21st century~\cite{SmaleI,SmaleII}, which asks whether the number of similarity classes of relative equilibria of $n$ positive point masses is finite for all $n$. In the broader sense of central configurations of arbitrary dimension, Smale's problem is part of the Chazy-Wintner-Smale finiteness conjecture~\cite{Albouy}. 
	
	The finiteness problem for central configurations has been solved only in the case $n=4$ by Hampton and Moeckel~\cite{Hampton-Moeckel,Hampton-MoeckelII}, who mainly employed the set of polynomial equations introduced by A. Albouy and A. Chenciner in the celebrated article~\cite{Albouy-Chenciner}. 
	A remarkable feature of the Albouy-Chenciner equations is that they depend only on the mutual distances of the configuration. Thus, besides eliminating the rotational and translational symmetries, the Albouy-Chenciner equations are at once valid for configurations of all dimensions. 
	Unlike the known deductions of the Albouy-Chenciner equations~\cite{Albouy-Chenciner,Hampton-Moeckel}, which use reduction by isometries, early on in this paper we obtain  a prototype of these equations for every system with $\mu_1=\longvec{O}$ as a nearly straightforward consequence of the basic interplay of the moments $\mu_0$, $\mu_1$ and $\mu_2$. Taking advantage of the broader perspective offered by the theory of moments, we also provide another set of algebraic equations exhibiting the properties of the Albouy-Chenciner equations mentioned in the previous paragraph. We call these new equations the \emph{extended Leibniz identities}, as they are closely related to a classical identity involving moments and mutual distances attributed to Leibniz. 
	

	
	
	In addition to providing a unified framework for the construction of families of equilibrium equations, we shall see that the theory of moments of weighted systems is a key tool to study the geometry of configurations.  
	
	In the previously mentioned classification of four-equal-mass central configurations, Albouy adopted and developed the results by O. Dziobek~\cite{Dziobek}, who used mutual distances as coordinates in order to describe $n$-body central configurations as critical points of a modified potential function subject to constraints that guaranteed the configuration was planar. The constraints used by Dziobek were expressed in terms of the determinants introduced by A. Cayley~\cite{Cayley}, which later became known as Cayley-Menger determinants~\cite{Berger,Blumenthal}. Dziobek's apparatus was adapted by Moeckel~\cite{MoeckelI} to dimension $n-2$, for arbitrary $n\geq 4$. 
	The critical point approach to central configurations motivated us to revisit the theory of constraints for mutual distances. The interplay of the first three moments allows us to establish, in a simple manner, the theory of constraints for mutual distances between the points of a configuration of prescribed dimension, or lying on a sphere. 
	
	

	Before moving on to a summary of each section in the paper, we make a brief note related to the moment problem mentioned in the first paragraph. We notice that the moment problem suggests an inverse of the equilibrium problem treated in the present work, viz., assuming the configuration $X$ is known, one would be asked to find the interactions $w_x$, $x \in X$, which make $X$ an equilibrium. For this problem, our equilibrium equations  offer low-degree homogeneous relations that can be used to determine the $w_x$. 
	

	\subsection*{Outline of Specific Contents}   
	The fundamentals of the theory of moments of weighted systems are reviewed in section~\ref{secWS}, culminating with the Huygens-Leibniz-K\"{o}nig (HLK) theorem, which collects classical results relating $\mu_0, \mu_1$ and $\mu_2$.  We demonstrate three consequential corollaries of the HLK theorem with conditions on the second moment $\mu_2$ of a weighted system with $\mu_0=0$ to have $\mu_1=\longvec{O}$. 
	
	Equilibrium equations for systems of interacting particles are introduced in section~\ref{Statics}. 
	We apply the results of section~\ref{secWS} and section~\ref{matrix} to obtain equilibrium equations that are either linear or bilinear on the  $\varphi_{x,y}$, and linear on the squared mutual distances between points of $X$. These equations 
	are the generalization of the Albouy-Chenciner equations for central configurations as well as the set of equations we call the extended Leibniz identities. 
	
	
	
	In section~\ref{matrix} we use matrices to provide a more algebraic approach to the study of moments, which are viewed as linear maps on the space of weight functions. The kernels of these maps are related in a way precisely stated in theorem~\ref{bridge}, which serves as basis for the study of mutual distance constraints carried out in section~\ref{appendix}. The section closes with a proposition that justifies the extended Leibniz identities stated in section~\ref{Statics}.

	The study of weighted systems $(X,w)$ with $\mu_1=\longvec{O}$ calls for some additional concepts, specially that of \emph{codimension} of a configuration. In section~\ref{appendix} we use the remarkable (albeit simple) connection between identically vanishing first moment (and thus the notion of equilibrium!) and the codimension of a configuration. This connection and corollary~\ref{cospherical} produce the classical co-sphericity condition of Cayley~\cite{Cayley}. 
	The weight functions of systems with $\mu_1=\longvec{O}$ (and fixed $X$) form a vector space $\ssp_0(X)$ of dimension equal to the codimension $c:=n-d-1$ of $X$, where $n$ is (throughout this paper) the number of points in $X$, and $d$ is dimension of the affine closure of $X$, usually referred to as the \emph{dimension} of $X$. Theorem~\ref{bridge} implies a relation between the dimension of the configuration and vanishing Cayley-Menger determinants of subconfigurations of $X$ and, in the proof of theorem~\ref{constraints}, we indicate how to obtain the independent (Cayley-Menger determinantal) constraints that ensure a configuration has a given dimension. Theorem~\ref{constraints} provides the expression ${c+1 \choose 2 }$ for the number of constraints that ensure a configuration has dimension $d=n-c-1$. This expression generalizes the corresponding formul\ae~ for dimensions 1, 2 and 3 in Wintner~\cite{Wintner}, $\S357$. 
	

	Section~\ref{cc} concerns applications of the preceding sections to central configurations. We consider the not so explored problem of central configurations with zero total mass, for which some geometric properties and algebraic relations 
	are directly determined.
	For central configurations with nonzero total mass, the fact that $\text{dim } \ssp_0(X)=c$ is used to extend to arbitrary codimensions some determinantal identities obtained by Dziobek for central configurations with codimension one~\cite{Dziobek,Hagihara}.

	The article concludes with appendix A, where we apply the HLK theorem of section~\ref{secWS} to deduce simple linear constraints for squared mutual distances between points in space and the points of a fixed configuration.

		\section{Moments of a Weighted System} \label{secWS}
		
		We begin with an overview of affine geometry; see \cite{Berger} for a thorough exposition.  Let $\affspa$ be an \emph{affine space over a field} $K$, that is, a triple formed by a point set $\mathbb{A}$, a vector space $\longvec{\mathbb{A}}$ over $K$ and a free action $\tau$ of the additive group of $\longvec{\mathbb{A}}$ on $\mathbb{A}$. The vector $\longvec{v} \in \longvec{\mathbb{A}}$ acts on the point $x \in \mathbb{A}$ by translating it to a point $y \in \mathbb{A}$ and is unique in that sense, so we write
		\[ \overrightarrow{xy}=\longvec{v}.\]
		The zero vector, $\longvec{O}$, corresponds to the trivial translation. For any pair of points $x,y \in \mathbb{A}$, it is assumed the existence of a vector $\longvec{v} \in \longvec{\mathbb{A}}$ satisfying the above equation; in other words, $\tau$ is transitive. It is common to write 
		\[ x+\longvec{xy}=y \quad \text{or} \quad \longvec{xy}=y-x.\]
		The pair $(\longvec{\mathbb{A}},\tau)$ is the \emph{affine structure} of the affine space $\affspa$, which is often denoted simply by $\mathbb{A}$. If $\longvec{\mathbb{A}}$ is endowed with an inner product, we call $\mathbb{A}$ an \emph{Euclidean affine space} and use a $\cdot$ to indicate the inner product. A basic and very useful result is the \emph{Chasles relation}: $\longvec{xy}+\longvec{yz}=\longvec{xz}$.

		
		
		Let $X$ be a finite subset of an affine space $\mathbb{A}$ and $w$ be a function from $X$ to $K$, the \emph{field of scalars} of $\mathbb{A}$. We refer to $X$ as a \emph{configuration} of points in $\mathbb{A}$, and to $w$ as a \emph{weight function} on $X$.  Three key maps can be assigned to the \emph{weighted system} formed by $X$ and $w$, i.e., to the pair $(X,w)$.
		
		\begin{defn} \label{fundamental}
			The \emph{total weight} $\mu_0$, the \emph{first moment} $\mu_1:\mathbb{A}\longrightarrow \longvec{\mathbb{A}}$, and (for Euclidean affine spaces) the \emph{second moment} $\mu_2:\mathbb{A}\longrightarrow K$ of the weighted system $(X,w)$ are respectively given by
			\[ \mu_0=\sum_{x \in X} w(x), \ \ \ \mu_1(p)= \sum_{x \in X} w(x) \longvec{px}, \ \ \text{and} \ \ \mu_2(p)= \sum_{x \in X} w(x) \longvec{px}^2, \ \ \forall p \in \mathbb{A},\]
			where the square of a vector stands for the inner product of the vector with itself.
		\end{defn}
		
		
		The total weight can be thought of as the ``zeroth moment'', hence its notation. It is sometimes convenient to see $\mu_0,$ $\mu_1$ and $\mu_2$ as linear maps on the weight function (see section~\ref{matrix}).
		
		We state some basic results regarding the moments of a fixed weighted system $(X,w)$ in a given affine space $\mathbb{A}$.
		
		\begin{lemma} \label{firstmu}
			For all $p,q \in \mathbb{A}$, we have that $ \mu_1(p)-\mu_1(q)=\mu_0\longvec{pq}.$
		\end{lemma}
		
		The proof is straightforward.
		
		If $\mu_0 \neq 0$, there is a unique point $G \in \mathbb{A}$ such that $\mu_1(G)=\longvec{O}$, namely
		\[G=q+\frac{1}{\mu_0}\sum_{x \in X}w(x) \longvec{qx},\]
		where $q$ is an arbitrary point of $\mathbb{A}$. We call $G$ the \emph{barycenter} of $(X,w)$ and denote it by $\bary(X,w)$. 
		
		
		
		\begin{prop}  Let $p \in \mathbb{A}, x \in X$. \label{red1stmom}
			\begin{enumerate}
				\item If $\mu_0 \neq 0$ and $G=\text{\emph{bar}}(X,w)$, then $\mu_1(p)=\mu_0\longvec{pG}.$
				\item If $\mu_0=0$ and $w(x) \neq 0$, then $\mu_1(p)=w(x) \longvec{G_x x}$, where $G_x$ is the barycenter of the weighted system $\left(X\setminus\{x\}, w\big|_{X \setminus \{x\}}\right)$.
			\end{enumerate}
		\end{prop}
		\begin{proof} 
			Part (1) follows directly from lemma~\ref{firstmu} and the definition of barycenter. For part (2), note that $\mu_1$ is constant according to lemma~\ref{firstmu}, so we have that $\mu_1(p)=\mu_1(x)$. But the total weight of the weighted system $\left(X\setminus\{x\}, w\big|_{X \setminus \{x\}}\right)$ is $-w(x)\neq 0$, and the first moment of $\left(X\setminus\{x\}, w\big|_{X \setminus \{x\}}\right)$ at the point $x$ is $\mu_1(x)$. So, by applying part (1) to the latter system, with $p=x$, we obtain part (2).		
		\end{proof} 
		
		The following are immediate consequences of lemma~\ref{firstmu} and proposition~\ref{red1stmom}. 
		
		\begin{cor}\label{elegant} ${ }$
			\begin{itemize}
				\item[(a)] $\mu_0 \neq 0$ if and only if $\mu_1$ is a bijection.
				\item[(b)] $\mu_0=0$ if and only if $\mu_1$ is constant. 
			\end{itemize}
			In particular, $\mu_1=\longvec{O}$ implies $\mu_0=0$.
		\end{cor}
		
		\begin{cor} \label{barycentriccoords}
			If $\mu_1=\longvec{O}$, then for every $x \in X$ with $w(x)\neq 0$, we have that $x=\bary\left(X\setminus\{x\}, w\big|_{X \setminus \{x\}}\right)$.
		\end{cor}
		
		Notice that, if $X \setminus \{x\}$ in corollary~\ref{barycentriccoords} is a simplex with the same dimension as $X$, then $w\big|_{X \setminus \{x\}}$ defines homogeneous barycentric coordinates for $x$ with respect to $X \setminus \{x\}$. This remark is key to some relevant applications and will be explored in detail in a subsequent work.
		
		
		Regarding the second moment and its relation to the total weight and first moment, we have the following key theorem to which, for briefness, we shall refer by the acronym HLK.
		
		\begin{HLK} Let $\mathbb{A}$ be an Euclidean affine space, $(X,w)$ a weighted system in $\mathbb{A}$, and $\mu_0, \mu_1, \mu_2$ the moments of $(X,w)$. For every $p,q \in \mathbb{A}$, we have that
			\begin{equation} \label{leibniz} 
				\mu_2(p)-\mu_2(q)=\longvec{pq} \cdot (\mu_1(p)+\mu_1(q)).
			\end{equation}
			Therefore  
			\begin{enumerate}
				\item if $\mu_0 \neq 0$ and $G=\text{\emph{bar}}(X,w)$, then $\mu_2(p)-\mu_2(G)=\mu_0\, \longvec{pG}^2$, for all $p \in \mathbb{A}$;
				\item if $\mu_0=0$ and $\mu_1=\longvec{v}$, then $\mu_2(p)-\mu_2(q)=2 \longvec{v}\cdot \longvec{qp},$
				for all $p,q \in \mathbb{A}$; and
				\item for every $q \in \mathbb{A}$, the \emph{Leibniz's identity}
				\[ \mu_0 \mu_2(q)- \mu_1(q)^2=\sum_{\{x,y\} \subset X} w(x)w(y)\longvec{xy}^2
				\]
				holds.
			\end{enumerate}
		\end{HLK}
		\begin{proof}
			In order to obtain equation~\eqref{leibniz}, we just have to multiply both sides of the elementary identity
			\[ \longvec{px}^2= \longvec{qx}^2+\longvec{pq}^2+ 2 \longvec{qx} \cdot \longvec{pq} \]
			by $w(x)$, add for all $x \in X$, and apply lemma~\ref{firstmu}. Picking $q=G$ gives (1), and taking $q \in \mathbb{A}$ arbitrary produces (2). To prove (3), set $p=y \in X$. From equation~\eqref{leibniz}, we can write
			\[ \sum_{y \in X} w(y)\mu_2(y)=\mu_0\mu_2(q)-\mu_1(q)^2+\sum_{y \in X} w(y) \, \longvec{yq} \cdot \mu_1(y) .\]
			By lemma~\ref{firstmu}, we can replace $\mu_1(y)$ by $\mu_0 \longvec{yq} +\mu_1(q)$ in the latter summation. Leibniz's identity follows. 
		\end{proof}
		
		The HLK theorem has a broad range of applications. Below we state two of its corollaries, the second of which is consequential to a significant part of this work, including the formulation of equilibrium problems (section 3) and the study of configurations described in terms of mutual distances (section~\ref{appendix} and appendix A). 
		
		\begin{cor} \label{cute}
			We have that $\mu_1=\longvec{O}$ if and only if $\mu_0=0$ and
			\[\sum_{\{x,y\} \subset X} w(x)w(y)\longvec{xy}^2=0.\]
		\end{cor}
		
		\sloppy	
		\begin{cor} \label{strongconverse} Suppose $\mu_0=0$. The following assertions are equivalent:  
			\begin{itemize}
				\item[(i)] $\mu_1=\longvec{O}$; 
				\item[(ii)] $\mu_2$ is constant;
				\item[(iii)] the restriction of $\mu_2$ to $X$ is constant.
			\end{itemize}
			
		\end{cor}
		\begin{proof}
			If $\mu_1=\longvec{O}$, then $\mu_2$ is constant according to part (2) of the HLK theorem. Thus we have (i) $\Rightarrow$ (ii); (ii) $\Rightarrow$ (iii) is immediate. If the restriction of $\mu_2$ to $X$ is constant and $\mu_1=\longvec{v}$ then, from the HLK theorem, part (2),
			\[ 0=\mu_2(y)-\mu_2(x)=2\longvec{v}\cdot \longvec{xy}, \quad \text{for all} \ x,y \in X.\]
			But, for all $y \in X$,
			\[ \longvec{v}=\mu_1(y)=\sum_{x \in X} w(x) \longvec{xy} \in \vecaffX,\]
			where $\vecaffX$ is the subspace of $\longvec{\mathbb{A}}$ spanned by the vectors $\longvec{xy}$ with $x,y \in X$. So we must have that $\longvec{v}=\longvec{O}$, and (iii) $\Rightarrow$ (i).
		\end{proof} 
		
		Corollary~\ref{strongconverse} has the useful implication below. 
		\begin{cor} \label{albchengen}
			The weighted system $(X,w)$ has $\mu_1=\longvec{O}$ if and only if $\mu_0=0$ and 
			\begin{align} \label{sec-mom}
				&\sum_{z \in X} w(z) (\longvec{xz}^2-\longvec{yz}^2)=0 \quad \text{or} \quad  \sum_{z\in X\setminus\{x\} } w(z)(\longvec{xy}^2-\longvec{yz}^2+\longvec{zx}^2)=0,
			\end{align}
			holds for every $x,y \in X$.
		\end{cor}
		\begin{proof} Equations~\eqref{sec-mom} are mere restatements of assertion (iii) in corollary~\ref{strongconverse}. The equations are equivalent since
			\[ \sum_{z \in X} w(z) (\longvec{xz}^2-\longvec{yz}^2)=\sum_{z \in X\setminus\{x\}} w(z) (\longvec{xz}^2-\longvec{yz}^2)+\left(\sum_{z \in X\setminus\{x\}}w(z)\right)\longvec{xy}^2.\]
		\end{proof}
		
		\noindent The relationship between corollaries~\ref{cute} and~\ref{strongconverse} will be clarified in proposition~\ref{moreeqeqs}.

		\section{Equilibrium Equations}  \label{Statics}
		
		Let $X \subset \mathbb{A}$ represent a finite set of interacting particles. Let us assume that the force of interaction between two distinct particles $x,y\in X$ can be expressed as
		\begin{equation} \label{forces}
			\quad	 \longvec{F}_{x,y}=\varphi_{x,y}\longvec{xy}. 
		\end{equation}
		The coefficients $\varphi_{x,y}$ are (rather arbitrary) elements of $K$ and may, by allowing $X$ to vary, be functions of the positions of the points in $X$. 
		
		The \emph{total force} on particle $x$ is $\longvec{F}_x=\sum_{y \neq x}\longvec{F}_{x,y}$. As customary, we define an \emph{equilibrium} of a system of interacting particles as a configuration for which the total force on each particle is zero. For each $x \in X$, let $(X,w_x)$ be the weighted system defined by
		\begin{equation} \label{weights}
			w_x(y)=\varphi_{x,y}, \quad \text{if} \ y \neq x, \quad w_x(x)=-\sum_{y \neq x } \varphi_{x,y}. 
		\end{equation}
		It is clear that each system $(X,w_x)$ has zero total weight and (constant) first moment $\mu_1=\longvec{F}_x$. 
		The next result is a direct consequence of this key observation.
		
		\begin{thm} \label{equilibriathm}
			A configuration $X$ of interacting particles whose forces of interaction are of the form~\eqref{forces} is an equilibrium if and only if the corresponding weighted systems $(X,w_x)$, $x \in X$, defined by equations~\eqref{weights} have zero first moment.
		\end{thm}
		Theorem~\ref{equilibriathm} represents a landmark in our work as it bridges the theory of moments from section~\ref{secWS} with applications to equilibrium problems. For instance, corollary~\ref{albchengen}, theorem~\ref{equilibriathm} and proposition~\ref{moreeqeqs} allow us to write (homogeneous) equations for equilibria which depend only on the force coefficients $\varphi_{x,y}$ and squared mutual distances, and have a relatively simple form.
		

		\begin{cor} \label{genalbouychenciner}
			A configuration $X$ is an equilibrium of a system of particles in an Euclidean affine space which interact according to equations~\eqref{forces} if and only if there holds one of the following sets of identities 
		\begin{equation} \label{albouy-chenciner}
			\sum_{z \in X \setminus \{x\}} \varphi_{x,z}(\longvec{xy}^2-\longvec{yz}^2+\longvec{zx}^2)=0,
		\end{equation}
		or \begin{equation} \label{leibniz1}
			\sum_{\{u,z\} \subset X} \varphi_{x,u}\varphi_{y,z} \longvec{uz}^2=0. 
		\end{equation}
		for every pair of points $x,y \in X$.
	\end{cor} 
	\noindent As a matter of fact, corollary~\ref{cute} ensures that the converse statement in the above corollary is true as long as just the equations~\eqref{leibniz1} with $x=y$ hold. 
	
	Due to their relevance and relation to classical results, we refer to equations~\eqref{albouy-chenciner} and~\eqref{leibniz1} as the \emph{generalized Albouy-Chenciner equations} and \emph{extended Leibniz identities}, respectively. A somewhat curious fact regarding corollary~\ref{genalbouychenciner}  is that the Euclidean structure of the affine space need not coincide with the one used to define the distances that may appear in the arguments of $\varphi_{x,y}$.
	
	The following remarks contain comments on applications.
	
	\begin{rem}
		Expression~\eqref{forces} is commonly adopted to describe internal forces in Classical Mechanics, together with conditions on the $\varphi_{x,y}$ in order to comply with the homogeneity and isotropy of space~\cite{Arnold}. For our purposes, though, it is necessary neither to make these additional assumptions, nor to require the validity of Newton's third law of motion, which corresponds to equalities $\varphi_{x,y}=\varphi_{y,x}$ for all $x,y \in X$. 
	\end{rem}
	
	\begin{rem} \label{scope}
		If a function $f$ depends solely on the mutual distances between particles, then the gradient of $f$ for each individual particle has the form~\eqref{forces}, so theorem~\ref{equilibriathm} and corollary~\ref{genalbouychenciner} can be applied to study the critical points of $f$. In Physics, a force field is called \emph{conservative} when it is the gradient of a (potential) function.  If each $\varphi_{x,y}$ depends only on the distance from $x$ to $y$, that is, if the force fields $F_x:\longvec{F}_{x,y}$, $x,y \in X$ are \emph{central} according to the definition in~\cite{Arnold}, then $F_x$ is shown to be conservative. In addition, if all the forces derive from a single potential function, a simple calculation shows that Newton's third law of motion is valid for the corresponding system. 
	\end{rem}
	
	\begin{rem} \label{inertial}
		External forces acting on the particles (due to a moving reference frame, for instance) may be dealt with by the above method as long as the total force on each particle remains the summation of terms of the form~\eqref{forces}. The example of central configurations of nonzero masses and nonzero total mass illustrates this possibility, see section~\ref{cc}, equations~\eqref{cceqalt} and~\eqref{ccweight} (notice that there, in view of Newton's second law of motion, we found it more convenient to work with accelerations instead of forces). 
		Reference~\cite{Arustamyan} discusses related problems.
	\end{rem}	  
	
	\section{Matrix Formulations} \label{matrix}
	In the next paragraphs we employ basic Linear Algebra to deepen our analysis of the relationship between the moments $\mu_0,$ $\mu_1$ and $\mu_2$. The results of this section yield the equilibrium equations~\eqref{leibniz1} in section~\ref{Statics}, and will be used in section~\ref{appendix} to study the geometry of configurations.
	
	For convenience, from now on we let $X=\{x_1,\ldots,x_n\}$.  
	A weight function $w:X\longrightarrow K$ shall be represented by the \emph{weight vector} $W \in K^n$ whose coordinates are $w(x_j)=w_j$, $j=1,\ldots,n$.

	We select a reference frame and introduce affine coordinates in the affine closure of $X$, $\aff(X)$, in the usual way.  Let $\mathbb{W}(X)=K^n$ denote the space of weight vectors defined on $X$. We call $\ssp(X)$ the \emph{weight space} of $X$. Let us define the $1 \times n$ matrix ${\bm{1}}=\left[1 \ldots 1 \right]$, and the matrices
	\[  \mathcal{X}=\left[x_{ij} \right]_{1\leq j\leq n} \quad \text {and} \quad \mathcal{B}(X)=\left[s_{ij} \right]_{1\leq i,j \leq n},\]
	where $x_{ij}$ is the $i^{\text{th}}$ coordinate of $x_j$ and $s_{ij}=\longvec{x_ix_j}^2$. The matrices ${\bm 1}$, $\mathcal{X}$ and $\mathcal{B}(X)$ establish linear maps on $\mathbb{W}(X)$ corresponding, respectively, to the moments $\mu_0, \mu_1$ and $\mu_2$.
	
	Notice that the condition $\mu_0=0$ corresponds to the weight vector $W$ belonging to the kernel of $\bm{1}$. Denote $\text{Ker}({\bm 1})$ by $\mathbb{D}$. The condition $\mu_1=\longvec{O}$ corresponds to $\bm{1}W=0$ and $\mathcal{X} W=\bm{0}^T$, where $\bm{0}$ is the $1 \times n$ matrix of zeros. So $W$ is a weight vector in the kernel of the \emph{augmented configuration matrix} 
	\[ \mathcal{X}_a=\left[\begin{array}{c}
		\bm{1} \\
		\mathcal{X}
	\end{array}
	\right].
	\]
	Henceforth we write $\text{Ker}(\mathcal{X}_{a})=\ssp_0(X)$. It is clear that $\ssp_0(X) \leq \mathbb{D}$. Now consider $\mathcal{B}(X)$ as a map on $\mathbb{D}$. According to corollary~\ref{strongconverse}, equivalence between sentences (i) and (iii), we have that
	\[ W \in \ssp_0(X) \quad \Longleftrightarrow \quad \exists w_0 \in K \ \ \text{such that} \ \ \mathcal{B}(X)W=w_0 \bm{1}^T.\]
	So  $-w_0 {\bm 1}^T+\mathcal{B}(X)W={\bm 0}^T$, and we may write
	\begin{equation} \label{Cayley-Menger}
		w \in \ssp_0(X) \quad \Longleftrightarrow \quad \mathcal{C}(X) \widetilde{W}={\bm 0},
	\end{equation}
	where $\widetilde{W}=(-w_0,w_1,\ldots,w_n)$ and    
	\[ \mathcal{C}(X)=\left[\begin{array}{cc}
		0&{\bm{1}} \\
		{\bm{1}}^T&\mathcal{B}(X)
	\end{array}
	\right],
	\]
	The determinant of $\mathcal{C}(X)$ is known as the \emph{Cayley-Menger determinant} of $X$~\cite{Berger,Blumenthal}.
	
	\begin{thm} \label{bridge}
		For any configuration $X$, we have that $\emph{Ker}\,\, \mathcal{B}(X) \cap \mathbb{D} \leq \ssp_0(X) \cong \emph{Ker}\, \, \mathcal{C}(X)$, with $\emph{Ker}\,\,\mathcal{B}(X) \cap \mathbb{D} = \ssp_0(X)$ if and only if for every weighted system $(X,w)$ with $\mu_0=0$ and $\mu_1=\longvec{O}$, we also have $\mu_2=0$.
	\end{thm}  
	
	\begin{proof} 
		If $W \in \text{Ker} \,\, \mathcal{B}(x) \cap \mathbb{D}$, then the restriction of $\mu_2$ to $X$ vanishes identically. Thus $\mu_1=\longvec{O}$, by corollary~\ref{strongconverse}, so $\text{Ker} \,\, \mathcal{B}(x) \cap \mathbb{D} \leq \ssp_0(X)$. The isomorphism $\ssp_0(X) \cong \text{Ker}\, \, \mathcal{C}(X)$ is clear from~\eqref{Cayley-Menger}.
	\end{proof}
	\noindent We notice that the aforementioned natural isomorphism $W \mapsto \widetilde{W}$ in theorem~\ref{bridge} seems to have firstly appeared in corollary 3.4 of reference~\cite{Dias}.
	
	A configuration $X$ is called \emph{co-spherical configuration} if there exists a point $p \in \mathbb{A}$ and a scalar $R \in K$ such that $\longvec{px_i}^2=R^2$, for all $i=1,\ldots,n$.
	
	\begin{cor} \label{cospherical}
		If $X$ is a co-spherical configuration, then $\emph{Ker}\,\,\mathcal{B}(X) \cap \mathbb{D}=\ssp_0(X)$. 
	\end{cor}
	\begin{proof}
		Let $p\in \mathbb{A}$, $R \in K$ be as in the definition of co-spherical configuration, and let $W \in \ssp_0(X)$. Since $W \in \mathbb{D}$, we have that
		\[ \mu_2(p)=\sum_{i=1}^n w_i \longvec{px_i}^2=\mu_0 R^2=0,\]
		Thus $\mu_2$ is the null function, and $W \in\text{Ker}\,\, \mathcal{B}(X),$ by theorem~\ref{bridge}.    
	\end{proof}
	
	The matrix $\mathcal{B}(X)$ is known as the \emph{relative configuration matrix}. It is noteworthy that $\mathcal{B}(X)$ plays a key role in the reduction of the $n$-body problem by rotations and translations~\cite{Albouy-Chenciner,Moeckel}, where it is viewed as a symmetric bilinear form on $\mathbb{D}$. We verify the basic properties of this bilinear form below.
	
	\begin{prop} \label{moreeqeqs}
		Let $X$ be a configuration in an affine Euclidean space, and let $\mathcal{B}(X)$ be its relative configuration matrix viewed as a bilinear form on the space $\mathbb{D}$ of weight vectors with $\mu_0=0$.
		\begin{itemize}
			\item[(i)] The quadratic form defined by $\mathcal{B}(X)$ is negative semi-definite.
			\item[(ii)] The kernel of $\mathcal{B}(X)$ is $\ssp_0(X)$. In particular, the restriction of $\mathcal{B}(X)$ to $\ssp_0(X)$ is the null form.
		\end{itemize}
		
	\end{prop}
	\begin{proof} Item (i) is a direct consequence of the Leibniz identity, see the HLK theorem, part (3).
		
		Notice item (ii) follows directly from corollary~\ref{cute}. We provide an alternative proof that helps to clarify the relationship between corollaries~\ref{cute} and~\ref{strongconverse}. Consider the
		basis of $\mathbb{D}$ consisting of the weight vectors $E_{12}=(1,-1,0,\ldots,0), \ldots, E_{1n}=(1,0,\ldots,0,-1)$. By corollary~\ref{strongconverse}, $W \in \ssp_0(X)$ if and only if $\mu_2(x_1)=\mu_2(x_i)$, $i=1,\ldots,n,$ that is, if and only if
		\[ E_{1i}^T \mathcal{B}(X)W=0, \quad \text{for all} \quad i=2,\ldots,n,\]
		which is equivalent to $W^T\mathcal{B}(X)={\bm 0}$ on $\mathbb{D}$. 
	\end{proof}
	
	\begin{rem}
		Early on in the next section we shall define the dimension of a configuration. Proposition~\ref{moreeqeqs}(ii) and lemma~\ref{dimnull} imply that the rank of $\mathcal{B}(X)$ coincides with the dimension of $X$. This fact and proposition~\ref{moreeqeqs}(i) are consistent with $X$ being irreducibly isometrically embedded in $\mathbb{A}$ (see the corollary on page 107 of~\cite{Blumenthal}).
	\end{rem}


	\section{Constraints for Mutual Distances} \label{appendix}
	
	In this section we draw some consequences of theorem~\ref{bridge} for describing configurations in terms of mutual distances. 
	
	As in section~\ref{matrix}, let $X=\{x_1,\ldots,x_n\}$. Denote the dimension of the affine closure of $X$ by $\text{dim}(X)$, or just $d$, and call it the \emph{dimension} of $X$. The difference $c$ between the maximum possible value for $\text{dim}(X)$, which is $n-1$, and $d$, that is $c=(n-1)-d$, is called the \emph{codimension} of $X$, and we write $\text{codim}(X)=c$. 
	

	In the classical paper~\cite{Cayley},  Cayley established determinantal conditions for a configuration to have codimension one, and for  it being co-spherical. Cayley's conditions will be easily deduced using the matrix framework from section~\ref{matrix}. The main result of this section is a theorem which exhibits a complete set of independent constraints on mutual distances to ensure that a configuration has a given dimension. Constraints are also discussed in appendix~\ref{appendix+}, where we show how to obtain simpler, non-determinantal constraints using the HLK theorem.


	\begin{lemma} \label{dimnull}
		$\dim\mathbb{W}_0(X)=\text{\emph{codim}}\, (X)$.
	\end{lemma} 
	\begin{proof}
		The augmented configuration matrix $\mathcal{X}_a$ has the same rank as the matrix
		\[ \left[\begin{array}{cccc}
			1& 0& \cdots& 0 \\
			x_1& x_2-x_1 &\cdots & x_n-x_1
		\end{array}
		\right] ,
		\]
		that is, $d+1$. The nullity of $\mathcal{X}_a$ is thus $n-(d+1)=(n-1)-d=c$. This is the dimension of $\mathbb{W}_0(X)$.
	\end{proof}
	
	Notice that every configuration of codimension $c=0$ is a simplex, and so is co-spherical. Corollary~\ref{cospherical} implies the following necessary condition for co-sphericity when the codimension is positive.
	
	\begin{cor}
		If $X$ has codimension $c>0$ and is co-spherical, then $\dim \text{\emph Ker}\,\, \mathcal{B}(X)\geq c$ and, in particular, $\det \mathcal{B}(X)=0$.
	\end{cor}
	Cayley deduced the condition $\det \mathcal{B}(X)=0$ in~\cite{Cayley}. He showed that, in the case of four points on a circle, that is, if $n=4$ and $c=1$, this condition implies the classical Ptolemy identity. It turns out that $\det \mathcal{B}(X)=0$ is also a sufficient condition for co-sphericity (cf.~\cite{Berger}, chapter 9).
	
	Next we study the dimension of $X$. Borrowing the terminology from Celestial Mechanics, a configuration of codimension one will be designated as \emph{Dziobek configuration}, or \emph{Dziobek subconfiguration}, when it is a proper subset of a larger configuration. The statements in the following paragraph appear in similar form in reference~\cite{Dias}.
	
	Theorem~\ref{bridge} and lemma~\ref{dimnull} imply that $\mathcal{C}(X)$ has nullity equal to $c$. Thus, for $c=0$, the rank of $\mathcal{C}(X)$ is $n+1$, so $\det \mathcal{C}(X)\neq 0$. This means that there are no constraints on mutual distances, as expected. For configurations of codimension $c>0$, $\mathcal{C}(X)$ has rank $n-c+1$, so all subdeterminants of $\det \mathcal{C}(X)$ of size $(n-c+2)\times(n-c+2)$ must vanish.  In particular, the $(n-c+2)\times(n-c+2)$ principal minors of $\mathcal{C}(X)$, which are precisely the Cayley-Menger determinants of the Dziobek subconfigurations of $X$, are all equal to zero. We shall soon check that not all Cayley-Menger determinants of Dziobek subconfigurations are independent and that, in general, a much smaller number of such determinants suffices to guarantee the correct codimension for $X$.

	We firstly provide a heuristic argument for the number of constraints that must be satisfied by the mutual distances of a configuration of dimension $d$ formed by $n$ points in Euclidean space. Let $E(d)$ be the group of Euclidean (rigid) motions in $\mathbb{E}^d$. The manifold of congruence classes of configurations of $n$ points in $\mathbb{E}^d$ modulo isometries is $\mathbb{E}^{dn}\big /E(d)$, which has dimension $dn-d(d+1)/2$. The natural coordinates to describe the congruence classes are the $n \choose 2$ mutual distances between the points. For configurations of codimension $c=(n-1)-d>0$, the number of mutual distances exceeds the dimension of the manifold of congruence classes by
	\begin{equation} \label{number}
		{n \choose 2} -\left[dn-\frac{d(d+1)}{2}\right]={c+1 \choose 2},
	\end{equation}
	which is the number of constraints that must be imposed on the mutual distances. 
	
	Now we supply a formal statement and proof. 
	
	\sloppy
	\begin{thm} \label{constraints} Let $X$ be a configuration of dimension $d<n-1$. The mutual distances of $X$ are such that
		${c+1 \choose 2}$ functionally independent Cayley-Menger determinants of subconfigurations of $X$ vanish, where $c=\emph{codim}(X)$.
	\end{thm}
	\fussy
	\begin{proof} Inspired by the argument in section 3.5 of~\cite{Hagihara}, we use induction on the codimension of $X$. 
		When $c=1$, there is the unique constraint $\det \mathcal{C}(X)=0$, so the formula is valid in this case. 
		Suppose the formula holds for a configuration $X$ of $n$ points and dimension $d$, where $d=n-(c+1)$. Adding the point $p$ to $X$ without changing the dimension, the new configuration will have codimension $n+1-(d+1)=n-d=c+1$. The additional constraints are the Cayley-Menger determinants of subconfigurations of $X\cup \{p\}$ formed by a selection of (any) $d$ points in $X$, the point $p$, and each of the remaining $n-d$ points of $X$, one at a time. Thus we get a total of
		\[ {c+1 \choose 2}+n-d={c+2 \choose 2}\]
		constraints, as claimed.
	\end{proof} 
	
	It should be noted that, in the proof of theorem~\ref{constraints}, it is possible to choose the subconfiguration of $X\cup \{p\}$ formed by the $d$ points of $X$ and $p$ as $d$-dimensional simplexes, so the additional contraints are expressed as Cayley-Menger determinants of Dziobek subconfigurations of $X\cup \{p\}$.

	\section{Some Applications to central configurations} \label{cc}
	
	We call a \emph{configuration of point masses}, or \emph{bodies}, in $\mathbb{E}^N$, any weighted system $(X,m)$ such that $m$ is a nonvanishing function traditionally referred to as \emph{mass}. As in the previous section, we denote the points of $X$ by $x_i$ and $m_i$ is the value of $m$ at $x_i$, $i=1,\ldots,n$. We may interpret $m$ as the weight vector $(m_1,\ldots,m_n) \in \mathbb{W}(X)=\reals^n$. Assuming the bodies interact under a homogeneous potential, the corresponding acceleration vectors have the form
	\begin{equation*} 
		\longvec{\gamma}_j = \sum_{i\neq j} m_i s_{ij}^{a} \longvec{x_ix_j}, \quad  s_{ij}=\longvec{x_ix_j}^2, \ \ a\in \mathbb{R},
	\end{equation*}
	according to Newton's second law of motion.
	
	Following Albouy~\cite{Albouy}, we say that $(X,m)$ is a \emph{central configuration} if there exist a vector $\longvec{\gamma}_{O}$, a point $x_{O}$ and a real number $\lambda$ such that 
	\begin{equation} \label{cceq}
		\longvec{\gamma}_j-\longvec{\gamma}_O=\lambda \longvec{x_Ox_j}, \quad j=1,\ldots,n.
	\end{equation}
	Albouy has shown that, if $a=0$, then any configuration is central, and if $a \neq 0$ and $d=n-1$, i.e, if $c=0$, then $X$ is a regular simplex for all values of the \emph{total mass} $\mu_0$. In the cases $\mu_0\neq 0$, this assertion follows directly from lemma~\ref{dimnull} and theorem~\ref{cczeroweights} below (as in~\cite{Albouy}, we tacitly assume that $\lambda/\mu_0>0$).
	
	From now on we suppose that $a \neq 0$.
	
	Multiplying both sides of~\eqref{cceq} by $m_j$ and summing over $j$, we obtain an equation relating the constants $\lambda, x_O, \longvec{\gamma}_O$ and the total mass, namely
	\begin{equation} \label{parrel}
		\lambda \mu_1(x_O)=-\mu_0 \longvec{\gamma}_O.
	\end{equation}
	
	
	
	\begin{prop} \label{zeromuzero}
		Let $(X,m)$ be a central configuration with $\mu_0=0$ and $\lambda \neq 0$. Let $X_j= X\setminus\{x_j\}$, $j=1,\ldots,n$. We have that
		
		\begin{itemize}
			\item[(1)] $(X,m)$ has $\mu_1=\longvec{O}$.
			\item[(2)] $x_j=\text{\emph{bar}}\left(X_j, m\big|_{X_j}\right)$, for every $j=1,\ldots,n$.
			\item[(3)] $X$ has codimension $c\geq 1$ and $\mu_2$ is constant. In particular, we have the Leibniz identity $\sum_{i<j} m_i m_j s_{ij}=0$, and the equations
			\[  \sum_{k=1 }^n m_k(s_{ij}-s_{jk}+s_{ik})=0, \quad \text{for every } \ i,j=1,\ldots,n. \]
			If $c=1$ then, for every $j=1,\ldots,n$, $X_j$ is a $d$-dimensional simplex and $x_j$ is not on the boundary of the convex hull of $X_j$.

			
		\end{itemize}
	\end{prop}
	
	\begin{proof} Part (1) follows from equation~\eqref{parrel}, and (2) is a direct consequence of corollary~\ref{barycentriccoords}, because $m$ is nonvanishing. In order to prove (3), note that $\mu_1=\longvec{O}$ provides a nontrivial linear relation for the elements of $X$. Hence $c \geq 1$. Moreover, $\mu_2$ is constant according to corollary~\ref{strongconverse}. The equations follow from corollaries~\ref{cute} and~\ref{albchengen}. The assertion on the case $c=1$ follows directly from (2) and the barycentric interpretation of the (nonzero) coordinates of the mass vector. 
	\end{proof}
	For $n=4$, a full classification of a special family of central configurations with $\mu_0=0$, including examples of central configurations with $\lambda \neq 0$, can be found in~\cite{Celli}.

	Next, let us suppose that that the total mass is nonzero. As proved in~\cite{Albouy}, $(X,m)$ is a central configuration if and only if $\longvec{\gamma}_O$ is the null vector and $x_O=\bary(X,m)$ is the \emph{center of mass} of $(X,m)$. The central configuration equations can thus be put in the form 
	\begin{equation} \label{cceqalt}
		\sum_{i \neq j} m_i S_{ij} \longvec{x_jx_i}=\longvec{O}, \quad j=1, \ldots,n,
	\end{equation}
	where $S_{ij}=s_{ij}^{a}-\frac{\lambda}{\mu_0}$. Equations~\eqref{cceqalt} imply that, for each $j$, a weighted system $(X,C_j)$ with $\mu_1=\longvec{O}$ can be defined, namely
	\begin{equation} \label{ccweight}
		C_j(x_i)=m_i S_{ij}, \quad j \neq i,\ \  C_j(x_j)=-\sum_{i \neq j} m_iS_{ij},\ \ \text{for} \ \ j=1,\ldots,n.
	\end{equation}
	For convenience, we write $C_j(x_j)=m_j S_{jj}$. We have just proved the next theorem.
	
	\begin{thm} \label{cczeroweights}
		A configuration of particles $(X,m)$ with $\mu_0 \neq 0$ is a central configuration if and only if each weighted systems $(X,C_j)$, with $C_j$ defined by equations~\eqref{ccweight}, $j=1,\ldots,n$, have zero first moment.
	\end{thm}
	Therefore we can apply the results of sections~\ref{Statics} and~\ref{appendix} to the study of central configurations with $\mu_0 \neq 0$. In particular, we recover the important result below. 
	
	
	\begin{ACF} \label{AlbCheneqs}
		A configuration of particles $(X,m)$ with $\mu_0 \neq 0$ is central if and only if, for each $j,k=1,\ldots,n$, there holds 
		\[\sum_{ i \neq j } m_i S_{ij}(s_{ij}+s_{jk}-s_{ki})=0.\]
		
	\end{ACF}
	
	Corollary~\ref{genalbouychenciner}, lemma~\ref{dimnull}, theorem~\ref{cczeroweights} and equations~\eqref{ccweight} imply the following corollaries.
	\begin{cor} \label{leibnitzcceqs}
		A configuration of particles $(X,m)$ with $\mu_0 \neq 0$ is a central configuration if and only if
		\begin{equation*} \label{lebnizccs}
			\sum_{i<k} m_im_k S_{ij}S_{kl}s_{ik}=0, \quad j,l=1,\cdots,n.
		\end{equation*}
	\end{cor}
	\begin{proof}
		Apply proposition~\ref{moreeqeqs} to $W_1=C_j$ and $W_2=C_l$. For the converse statement, use corollary~\ref{cute} (as in corollary~\ref{genalbouychenciner}, just the equations with $j=l$ suffice).
	\end{proof}
	
	
	
	
	\begin{cor}  \label{nice1}
		Let $(X,m)$ be a central configuration with $\mu_0 \neq 0$ and codimension $c$. Any $(c+1) \times (c+1)$-minor of the $n \times n$ matrix with entries $S_{ij}$, $i,j=1,\ldots,n,$ is zero.
	\end{cor}
	\noindent The above statement follows lemma~\ref{dimnull} and the fact that $m$ is nonvanishing. In the Dziobek case, $c=1$, we have that $S_{ij}S_{kl}-S_{il}S_{jk}=0,$
	for any four indices. When the indices are distinct, these identities are the well known equations (18) in~\cite{MoeckelI}. 
	



		\appendix
		\numberwithin{thm}{section}
		\numberwithin{equation}{section}
		\section{Direct Constraints for Mutual Distances}
		\label{appendix+}
		
		We provide a method for obtaining constraints for the distances from the points in a fixed configuration to a point in space. This method produces simpler alternatives to the expressions resulting from computing Cayley-Menger determinants. 
		
		\begin{prop} \label{ngons}
			Let $(X,w)$ and $(X',w')$ be weighted systems in an Euclidean affine space whose total weights are nonzero and equal, and whose barycenters coincide. The difference between
			the respective second moments has the same value at all points, namely, the value it has at the (common) barycenter. 
		\end{prop}
		\begin{proof} 
			We just have to apply part (1) of the HLK theorem to $(X,w)$ and $(X',w')$, subtract the corresponding sides of the resulting equations and do a simple manipulation.
		\end{proof}
		\noindent Proposition~\ref{ngons} 
		implies the following curious fact: the difference between the sums of the squared distances from any point to the vertices of two concentric regular $n$-gons of respective squared radii $s$ and $s'$ is always equal to $n(s-s')$, that is, it depends neither on the point nor on the relative position of the $n$-gons! As an illustration, consider a rhombus of squared semidiagonals $s$ and $s'$. The squared distances $s_i$, $i=1,\ldots,4$, in cyclic order, from the vertices of the rhombus to any point in space satisfy the constraint
		\begin{equation} \label{rhombusrel}
			s_1-s_2+s_3-s_4=\pm 2(s-s').
		\end{equation}
		The sign of the expression on the right-hand side of equation~\eqref{rhombusrel} is determined by computing the left-hand side at the centroid of the rhombus. 

		\vspace{.5cm}

		\noindent {\bf Acknowledgements.} I wish to thank Alain Albouy for early discussions on the project that lead to this work,  and to express my gratitude for the comments and suggestions received from the participants of the Mathematical Physics Seminar in Recife, and of events organized by Luis Fernando Mello, Antonio Carlos Fernandes and Lei Zhao. I specially thank Samuel Adrian Antz for sending a list of corrections to a previous version of the manuscript.\\
		
		\noindent {\bf Data availability statement.} The author declares that the data supporting the findings of this study are available within the paper.\\
		
		\noindent {\bf Funding/Competing interests statement.} The author has no relevant financial or non-financial interests to disclose. No support was received from any organization for the submitted work.


	\end{document}